\documentclass[11pt,onecolumn,draftclsnofoot]{IEEEtran}
\usepackage{fancyhdr}
\usepackage{amsmath,epsfig}
\usepackage{threeparttable}
\usepackage{epsf,epsfig}
\usepackage{amsmath}
\usepackage{amssymb}
\usepackage{amsfonts}
\usepackage[noadjust]{cite}
\usepackage{dsfont}
\usepackage{subfigure}
\usepackage{url}

\renewcommand{\baselinestretch}{1.4}

\newcommand{\tPout}{\tilde{P}_{\mathsf{out}}}
\newcommand{\hPout}{\hat{P}_{\mathsf{out}}}
\def\[{\left[}
\def\]{\right]}

\begin{document}

\newtheorem{theorem}{Theorem}
\newtheorem{acknowledgement}[theorem]{Acknowledgement}
\newtheorem{axiom}[theorem]{Axiom}
\newtheorem{case}[theorem]{Case}
\newtheorem{claim}[theorem]{Claim}
\newtheorem{conclusion}[theorem]{Conclusion}
\newtheorem{condition}[theorem]{Condition}
\newtheorem{conjecture}[theorem]{Conjecture}
\newtheorem{criterion}[theorem]{Criterion}
\newtheorem{definition}[theorem]{Definition}
\newtheorem{example}[theorem]{Example}
\newtheorem{exercise}[theorem]{Exercise}
\newtheorem{lemma}{Lemma}
\newtheorem{corollary}{Corollary}
\newtheorem{notation}[theorem]{Notation}
\newtheorem{problem}[theorem]{Problem}
\newtheorem{proposition}{Proposition}
\newtheorem{solution}[theorem]{Solution}
\newtheorem{summary}[theorem]{Summary}
\newtheorem{assumption}{Assumption}
\newtheorem{examp}{\bf Example}
\newtheorem{probform}{\bf Problem}
\def\remark{{\noindent \bf Remark:\hspace{0.5em}}}

\def\qed{$\Box$}
\def\QED{\mbox{\phantom{m}}\nolinebreak\hfill$\,\Box$}
\def\proof{\noindent{\emph{Proof:} }}
\def\poof{\noindent{\emph{Sketch of Proof:} }}
\def
\endproof{\hspace*{\fill}~\qed
\par
\endtrivlist\unskip}
\def\endproof{\hspace*{\fill}~\qed\par\endtrivlist\vskip3pt}

\def\E{\mathbb{E}}
\def\eps{\varepsilon}
\def\phi{\varphi}
\def\Lsp{{\boldsymbol L}}
\def\Bsp{{\boldsymbol B}}
\def\lsp{{\boldsymbol\ell}}
\def\Ltsp{{\Lsp^2}}
\def\Lpsp{{\Lsp^p}}
\def\Linsp{{\Lsp^{\infty}}}
\def\LtR{{\Lsp^2(\Rst)}}
\def\ltZ{{\lsp^2(\Zst)}}
\def\ltsp{{\lsp^2}}
\def\ltZt{{\lsp^2(\Zst^{2})}}
\def\ninN{{n{\in}\Nst}}
\def\oh{{\frac{1}{2}}}
\def\grass{{\cal G}}
\def\ord{{\cal O}}
\def\dist{{d_G}}
\def\conj#1{{\overline#1}}
\def\ntoinf{{n \rightarrow \infty }}
\def\toinf{{\rightarrow \infty }}
\def\tozero{{\rightarrow 0 }}
\def\trace{{\operatorname{trace}}}
\def\ord{{\cal O}}
\def\UU{{\cal U}}
\def\rank{{\operatorname{rank}}}
\def\acos{{\operatorname{acos}}}

\def\SINR{\mathsf{SINR}}
\def\SNR{\mathsf{SNR}}
\def\SIR{\mathsf{SIR}}
\def\tSIR{\widetilde{\mathsf{SIR}}}
\def\Ei{\mathsf{Ei}}
\def\l{\left}
\def\r{\right}
\def\({\left(}
\def\){\right)}
\def\lb{\left\{}
\def\rb{\right\}}

\setcounter{page}{1}

\newcommand{\eref}[1]{(\ref{#1})}
\newcommand{\fig}[1]{Fig.\ \ref{#1}}

\def\bydef{:=}
\def\ba{{\mathbf{a}}}
\def\bb{{\mathbf{b}}}
\def\bc{{\mathbf{c}}}
\def\bd{{\mathbf{d}}}
\def\bee{{\mathbf{e}}}
\def\bff{{\mathbf{f}}}
\def\bg{{\mathbf{g}}}
\def\bh{{\mathbf{h}}}
\def\bi{{\mathbf{i}}}
\def\bj{{\mathbf{j}}}
\def\bk{{\mathbf{k}}}
\def\bl{{\mathbf{l}}}
\def\bm{{\mathbf{m}}}
\def\bn{{\mathbf{n}}}
\def\bo{{\mathbf{o}}}
\def\bp{{\mathbf{p}}}
\def\bq{{\mathbf{q}}}
\def\br{{\mathbf{r}}}
\def\bs{{\mathbf{s}}}
\def\bt{{\mathbf{t}}}
\def\bu{{\mathbf{u}}}
\def\bv{{\mathbf{v}}}
\def\bw{{\mathbf{w}}}
\def\bx{{\mathbf{x}}}
\def\by{{\mathbf{y}}}
\def\bz{{\mathbf{z}}}
\def\b0{{\mathbf{0}}}

\def\bA{{\mathbf{A}}}
\def\bB{{\mathbf{B}}}
\def\bC{{\mathbf{C}}}
\def\bD{{\mathbf{D}}}
\def\bE{{\mathbf{E}}}
\def\bF{{\mathbf{F}}}
\def\bG{{\mathbf{G}}}
\def\bH{{\mathbf{H}}}
\def\bI{{\mathbf{I}}}
\def\bJ{{\mathbf{J}}}
\def\bK{{\mathbf{K}}}
\def\bL{{\mathbf{L}}}
\def\bM{{\mathbf{M}}}
\def\bN{{\mathbf{N}}}
\def\bO{{\mathbf{O}}}
\def\bP{{\mathbf{P}}}
\def\bQ{{\mathbf{Q}}}
\def\bR{{\mathbf{R}}}
\def\bS{{\mathbf{S}}}
\def\bT{{\mathbf{T}}}
\def\bU{{\mathbf{U}}}
\def\bV{{\mathbf{V}}}
\def\bW{{\mathbf{W}}}
\def\bX{{\mathbf{X}}}
\def\bY{{\mathbf{Y}}}
\def\bZ{{\mathbf{Z}}}

\def\mA{{\mathbb{A}}}
\def\mB{{\mathbb{B}}}
\def\mC{{\mathbb{C}}}
\def\mD{{\mathbb{D}}}
\def\mE{{\mathbb{E}}}
\def\mF{{\mathbb{F}}}
\def\mG{{\mathbb{G}}}
\def\mH{{\mathbb{H}}}
\def\mI{{\mathbb{I}}}
\def\mJ{{\mathbb{J}}}
\def\mK{{\mathbb{K}}}
\def\mL{{\mathbb{L}}}
\def\mM{{\mathbb{M}}}
\def\mN{{\mathbb{N}}}
\def\mO{{\mathbb{O}}}
\def\mP{{\mathbb{P}}}
\def\mQ{{\mathbb{Q}}}
\def\mR{{\mathbb{R}}}
\def\mS{{\mathbb{S}}}
\def\mT{{\mathbb{T}}}
\def\mU{{\mathbb{U}}}
\def\mV{{\mathbb{V}}}
\def\mW{{\mathbb{W}}}
\def\mX{{\mathbb{X}}}
\def\mY{{\mathbb{Y}}}
\def\mZ{{\mathbb{Z}}}

\def\cA{\mathcal{A}}
\def\cB{\mathcal{B}}
\def\cC{\mathcal{C}}
\def\cD{\mathcal{D}}
\def\cE{\mathcal{E}}
\def\cF{\mathcal{F}}
\def\cG{\mathcal{G}}
\def\cH{\mathcal{H}}
\def\cI{\mathcal{I}}
\def\cJ{\mathcal{J}}
\def\cK{\mathcal{K}}
\def\cL{\mathcal{L}}
\def\cM{\mathcal{M}}
\def\cN{\mathcal{N}}
\def\cO{\mathcal{O}}
\def\cP{\mathcal{P}}
\def\cQ{\mathcal{Q}}
\def\cR{\mathcal{R}}
\def\cS{\mathcal{S}}
\def\cT{\mathcal{T}}
\def\cU{\mathcal{U}}
\def\cV{\mathcal{V}}
\def\cW{\mathcal{W}}
\def\cX{\mathcal{X}}
\def\cY{\mathcal{Y}}
\def\cZ{\mathcal{Z}}
\def\cd{\mathcal{d}}
\def\Mt{M_{t}}
\def\Mr{M_{r}}
\def\O{\Omega_{M_{t}}}
\newcommand{\figref}[1]{{Fig.}~\ref{#1}}
\newcommand{\tabref}[1]{{Table}~\ref{#1}}

\newcommand{\var}{\mathsf{var}}
\newcommand{\fb}{\tx{fb}}
\newcommand{\nf}{\tx{nf}}
\newcommand{\BC}{\tx{(bc)}}
\newcommand{\MAC}{\tx{(mac)}}
\newcommand{\Pout}{P_{\mathsf{out}}}
\newcommand{\nnn}{\nn\\}
\newcommand{\FB}{\tx{FB}}
\newcommand{\TX}{\tx{TX}}
\newcommand{\RX}{\tx{RX}}
\renewcommand{\mod}{\tx{mod}}
\newcommand{\m}[1]{\mathbf{#1}}
\newcommand{\td}[1]{\tilde{#1}}
\newcommand{\sbf}[1]{\scriptsize{\textbf{#1}}}
\newcommand{\stxt}[1]{\scriptsize{\textrm{#1}}}
\newcommand{\suml}[2]{\sum\limits_{#1}^{#2}}
\newcommand{\sumlk}{\sum\limits_{k=0}^{K-1}}
\newcommand{\eqhsp}{\hspace{10 pt}}
\newcommand{\tx}[1]{\texttt{#1}}
\newcommand{\Hz}{\ \tx{Hz}}
\newcommand{\sinc}{\tx{sinc}}
\newcommand{\tr}{\mathrm{tr}}
\newcommand{\diag}{\mathrm{diag}}
\newcommand{\MAI}{\tx{MAI}}
\newcommand{\ISI}{\tx{ISI}}
\newcommand{\IBI}{\tx{IBI}}
\newcommand{\CN}{\tx{CN}}
\newcommand{\CP}{\tx{CP}}
\newcommand{\ZP}{\tx{ZP}}
\newcommand{\ZF}{\tx{ZF}}
\newcommand{\SP}{\tx{SP}}
\newcommand{\MMSE}{\tx{MMSE}}
\newcommand{\MINF}{\tx{MINF}}
\newcommand{\RC}{\tx{MP}}
\newcommand{\MBER}{\tx{MBER}}
\newcommand{\MSNR}{\tx{MSNR}}
\newcommand{\MCAP}{\tx{MCAP}}
\newcommand{\vol}{\tx{vol}}
\newcommand{\ah}{\hat{g}}
\newcommand{\tg}{\tilde{g}}
\newcommand{\teta}{\tilde{\eta}}
\newcommand{\heta}{\hat{\eta}}
\newcommand{\uh}{\m{\hat{s}}}
\newcommand{\eh}{\m{\hat{\eta}}}
\newcommand{\hv}{\m{h}}
\newcommand{\hh}{\m{\hat{h}}}
\newcommand{\Po}{P_{\mathrm{out}}}
\newcommand{\Poh}{\hat{P}_{\mathrm{out}}}
\newcommand{\Ph}{\hat{\gamma}}
\newcommand{\mat}[1]{\begin{matrix}#1\end{matrix}}
\newcommand{\ud}{^{\dagger}}
\newcommand{\C}{\mathcal{C}}
\newcommand{\nn}{\nonumber}
\newcommand{\nInf}{U\rightarrow \infty}

\title{\huge \setlength{\baselineskip}{30pt} Spectrum Sharing Between Cellular and Mobile Ad Hoc Networks: Transmission-Capacity Trade-Off}

\author{Kaibin Huang, Vincent K. N. Lau, Yan Chen\thanks{\setlength{\baselineskip}{15pt}
K. Huang, V. K. N. Lau, and Y. Chen are with Department of Electronic and Computer Engineering, Hong Kong University of Science and Technology, Hong Kong. Email: khuang@ieee.org, eeknlau@ust.hk, yanchen@ust.hk. Y. Chen is also affiliated with Institute of Information \& Communication Engineering, Zhejiang University, Hangzhou, 310027, P.R. China.
}\vspace{-50pt}}

\maketitle

\begin{abstract}\vspace{-10pt}
Spectrum sharing between wireless networks improves the efficiency of spectrum usage, and thereby alleviates spectrum scarcity due to growing demands for wireless broadband access.  To improve the usual underutilization  of the cellular uplink spectrum,  this paper studies spectrum sharing between a cellular uplink and a mobile ad hoc networks. These networks access either all frequency sub-channels or their disjoint sub-sets, called \emph{spectrum underlay} and \emph{spectrum overlay}, respectively. Given these spectrum sharing methods, the capacity trade-off between the coexisting networks is analyzed based on the {\it transmission capacity} of a network with Poisson distributed transmitters. This metric is defined as  the maximum density of transmitters subject to an outage constraint for a given signal-to-interference ratio (SIR). Using tools from stochastic geometry,  the transmission-capacity trade-off between the coexisting networks is analyzed, where both spectrum overlay and underlay as well as successive interference cancelation (SIC) are considered. In particular, for small target outage probability, the transmission capacities of the coexisting networks are proved to satisfy a linear equation, whose coefficients depend on the spectrum sharing method and whether SIC is applied. This linear equation shows that spectrum overlay is more efficient than spectrum underlay. Furthermore, this result also provides insight into the effects of different network parameters on transmission capacities, including link diversity gains, transmission distances, and the base station density.  In particular, SIC is shown to increase transmission capacities of both coexisting networks by a linear factor, which depends on the interference-power threshold for qualifying canceled interferers.

\end{abstract}
\vspace{-10pt}
\begin{keywords}\vspace{-10pt}
Spatial reuse;  wireless networks; Poisson processes; spectrum sharing; interference cancellation
\end{keywords}
\vspace{-10pt}
\section{Introduction}\label{Section:Intro}
Despite spectrum scarcity, most licensed spectrum are underutilized according to Federal Communications Commission \cite{FCC:SpectrumPolicyTaskForce}. In particular, in existing cellular systems based on frequency division duplex (FDD) such as FDD UMTS \cite{UMTS}, equal bandwidths are allocated for uplink and downlink transmissions, even though the data traffic for downlink is much heavier than that for uplink \cite{Marques:OppUse3GUplink:2008, KimJeong:CapUnbalanceULDL:CDMA:2000}. Spectrum sharing between wireless networks improves spectrum utilization, and will be a key solution  for broadband access in next-generation wireless networks \cite{Akyildiz:NextGenDynamSpectAccessCognitiveRadioSurvey:2006}. This motivates the study in this paper on sharing uplink spectrum between a cellular network and a mobile ad hoc network (MANET), which are referred to as the \emph{coexisting networks}. A basic question is then how is the trade-off between the capacities of these networks.

We provide answers to this question  in terms of the {\it transmission capacities } of the coexisting networks consisting of Poisson distributed transmitters. By extending the definition in \cite{WeberAndrews:TransCapWlssAdHocNetwkOutage:2005}, this metric is defined as the maximum weighted sum of the transmitter densities  of the coexisting networks so that all links will  satisfy an outage probability constraint for a target signal-to-interference ratio (SIR), where the weights depend on the spectrum-sharing method. We derive the transmission-capacity trade-off between the networks for different spectrum-sharing methods. Such results are useful for controlling the sizes of the coexisting networks for optimizing  uplink spectrum usage.

\subsection{Related Work and Motivation}
A  spectrum band can be either \emph{licensed}  or \emph{unlicensed}, where a license gives a network the exclusive right of spectrum usage. Depending on whether holding a licence, a wireless network is referred to as the \emph{primary} (e.g. cellular networks) or \emph{secondary} network (e.g. MANETs). Accessing a licensed band, the transmitters in a secondary network, called \emph{secondary transmitters},  must not cause significant interference to the receivers in the primary network, called \emph{primary receivers}.  One simple method of sharing a licensed band is to spread the signal energy radiated by each secondary transmitter over the whole band using spread spectrum techniques \cite{SimonBook:SpreadSpectrumHandbook:94}, suppressing the power spectrum density of the resultant interference to the primary receivers. This method is called  \emph{spectrum underlay} \cite{Zhao:SurveyDynamSpectAccess:2007, Akyildiz:NextGenDynamSpectAccessCognitiveRadioSurvey:2006, FCC:SpectrumPolicyTaskForce, MenonReed:OutageCompareUnderlayOverlaySpectrumSharing:2005}.

Another method for sharing licensed spectrum is called \emph{spectrum overlay}, where secondary transmitters access frequency sub-channels unused by nearby primary receivers. Recent research on spectrum overlay has been focusing on designing \emph{cognitive-radio} algorithms for secondary transmitters to opportunistically access the spectrum by exploiting the spatial and temporal traffic dynamic of the primary network \cite{Haykin:CognitiveRadio:2005, Zhao:SurveyDynamSpectAccess:2007, Akyildiz:NextGenDynamSpectAccessCognitiveRadioSurvey:2006}. Such algorithms require secondary transmitters to continuously detect  and track transmission opportunities by spectrum sensing, and decide on transmission based on sensing results  \cite{Zhao:CognitiveMACOppSpectrumAccess:POMDP:2007, Wild:DetectPrimRxCognitiveRadio:2005, Mishra:CoopSenseCogRadio:2006}. Such algorithms are vulnerable to sensing errors, and most important require complicated computation at the secondary transmitters, which usually have limited computational power.  For this reason, we consider the case where base stations in the cellular (primary) network coordinates spectrum overlay. Thereby ad hoc (secondary) transmitters use a simple random access  protocol rather than complicated cognitive-radio algorithms.

In unlicensed spectrum such as the \emph{industrial, scientific and medical} (ISM) bands, all networks have equal priorities for spectrum access. The networks using unlicensed bands include wireless local area networks (WLANs) \cite{IEEE802-11} and  wireless personal area networks (WPANs) \cite{IEEE802-15}. Due to mutual interference,  the coexistence of networks in the unlicensed bands degrades the networks' performance as shown by analysis \cite{Howitt:WLAN/WPANCoexist:2001, Han:OutageWPANWLANCoexistFading:2008}, simulation \cite{Thanthry:WLANAdHocNetworkCoexist:2004, Golmie:InterfEvalBluetoothWLAN:2003},  and measurement \cite{Angrisani:ExperimCoexistWLANWPAN:2008, Pollin:CoexistWPANWLANMeasurement:2008}. Sharing of unlicensed bands between competing networks is also studied using game theory \cite{Berlemann:RadioResourseSharingQosUnlicensedBand:2005, EtkinTse:SpectrumSharUnlicensedBands:2007}.

There exist few theoretical results on the network capacity trade-off between networks sharing spectrum despite this being a basic question. In \cite{ChanAndrews:UplinkCapInterfAvoidTwoTierFemtocell:2008}, the transmission capacities of a two-tier network are analyzed, which consists of a cellular network   and a network of femtocell hot-spots. In \cite{YinCui:ScalingLawsOverlaidWlssNetworks:2006}, the transport capacities\footnote{This metric introduced in \cite{GuptaKumar:CapWlssNetwk:2000} refers to end-to-end throughput per unit distance of a multi-hope wireless network.} of two coexisting multi-hop ad hoc networks are shown  to follow the optimum scaling laws for an asymptotically large number of network nodes. In  \cite{ChanAndrews:UplinkCapInterfAvoidTwoTierFemtocell:2008, YinCui:ScalingLawsOverlaidWlssNetworks:2006},  the network-capacity trade-off between coexisting networks is not analyzed.

The transmission capacity is used as the performance metric in this paper \cite{WeberAndrews:TransCapWlssAdHocNetwkOutage:2005}. Recently, this metric has been employed for analyzing different types of MANETs with Poisson distributed transmitters and an ALOHA-like medium-access-control layer, including spatial diversity \cite{WeberAndrews:TransCapAdHocNetwkDistSch:2006}, opportunistic transmissions \cite{WeberAndrews:TransCapAdHocNetwkDistSch:2006}, distributed scheduling \cite{HasanAndrews:GuardZoneAdHocNet:2007}, bandwidth partitioning \cite{JindalAndrews:BandwidthPartitioning:2007}, successive interference cancellation (SIC) \cite{WeberAndrews:TransCapWlssAdHocNetwkSIC:2005}, and spatial interference cancelation \cite{Huang:SpatialInterfCancel:2008} in MANETs.

\subsection{Contributions and Organization}
Our main contributions are summarized as follows.  The paper targets a cellular uplink network and a MANET sharing the uplink spectrum using either  spectrum overlay or underlay, where uplink users, base stations, and ad hoc transmitters all follow Poisson distributions but with different densities. Each transmitter modulates signals using frequency-hopping spread spectrum over the frequency sub-channels assigned to the corresponding network \cite{SimonBook:SpreadSpectrumHandbook:94}. First, considering an interference-limited environment, bounds on the SIR outage probabilities are derived for spectrum overlay and underlay with and without using SIC at receivers \cite{Andrews:InterfCancelCellularSys:ContempView:2005, WeberAndrews:TransCapWlssAdHocNetwkSIC:2005}.
Second, for small target outage probability, the transmission-capacities of the coexisting networks are showed to satisfy  a linear equation, whose coefficients depend on the overlay method and whether SIC is used.  Define the \emph{capacity region } as the set of feasible combinations of transmission capacities. Third, for small target outage probability, the capacity region for spectrum underlay is shown to be no larger than that for spectrum overlay. The former can be enlarged to be identical to the latter by choosing the transmission-power ratio between the two networks as derived. Finally, we characterize the effects of different parameters on transmission capacities of the coexisting networks. In particular, depending on whether using spectrum overlay or underlay, the transmission capacity of one or both networks grows linearly with the increasing base station density, linearly with the increasing spatial diversity gains raised to a fractional power, inversely with the decreasing distance between an ad hoc transmitter and its intended receiver. Moreover, SIC increases both transmission capacities by a linear factor that is a function of the interference-power threshold for qualifying canceled interferers.

Simulation results are also presented. As observed from these results, the derived bounds on outage probabilities are tight for different spectrum sharing methods with and without SIC. In particular, the outage probabilities converge to their lower bounds as the transmitter density decreases. Furthermore, the transmission capacity trade-off curves derived for asymptotically small target outage probabilities are observed to match simulation results in the non-asymptotic regime.

The remainder of this paper is organized as follows. The network and wireless channel models are described in Section~\ref{Section:NetworkModel}. In Section~\ref{Section:Outage}, the bounds on outage probabilities are  derived for different spectrum sharing methods. For small target outage probability, the  transmission-capacity trade-off is analyzed in Section~\ref{Section:TxCap}. Numerical and simulation results are presented in Section~\ref{Section:Numerical}, followed by concluding remarks in Section~\ref{Section:Conlusion}.

\section{Network Model}\label{Section:NetworkModel}

\subsection{Network Architecture}\label{Section:Sys:NetwkArch}
The spectrum-sharing cellular and ad hoc networks, referred to simply as \emph{coexisting networks},  are illustrated in Fig.~\ref{Fig:Network}. Following \cite{Baccelli:AlohaProtocolMultihopMANET:2006, WeberAndrews:TransCapWlssAdHocNetwkOutage:2005, WeberAndrews:TransCapAdHocNetwkDistSch:2006}, the transmitters in the MANET are modeled as a Poisson point process (PPP) on the two-dimensional plane, denoted as $\tilde{\Pi}$ with the density $\tilde{\lambda}$. Each transmitter in the MANET is associated with a receiver located at a fixed distance denoted as $\tilde{d}$.\footnote{Consideration of the randomness in $\tilde{d}$ does not provide little insight. It is straightforward to extend the results in this paper to include the randomness in $\tilde{d}$.} The transmission power of transmitters  is assumed fixed and denoted as $\tilde{\rho}$.

For the cellular network, the base stations and uplink users are modeled as two independent homogeneous PPPs denoted as $\Omega$ and $\Pi$, respectively. Their corresponding densities are represented by $\lambda_b$ and $\lambda$. Let $B_n$, $U_m$, $D_{n,m}$  denote the two-dimensional coordinates of the $n$th base station, the $m$th uplink user, and their distance, respectively. Thus,  $D_{n,m} = |B_n -U_m|$.\footnote{The operator $|X|$ gives the Euclidean distance between $X$ and the origin if $X$ is two-dimensional coordinates, or the cardinality of $X$ if $X$ is a set.}  To enhance the long-term link reliability, each uplink user transmits to the nearest base station. Consequently, the cellular network forms a \emph{Poisson tessellation} of the two-dimensional plane and each cell is known as a \emph{Voronoi} cell \cite{Okabe00:SpatialTesse}. The uplink users in the cell served by the $m$th base station, denoted as $\mathcal{V}_m$, is given as \cite{Okabe00:SpatialTesse}
\begin{equation}
\mathcal{V}_m = \l\{U\in\Pi\left| |U - B_m| < |U - B|\ \forall \ B \in \Omega\backslash\{B_m\} \r.\r\}.
\end{equation}
Based on their distances from the serving base station, the users in each cell are separated into \emph{inner-cell} and \emph{cell-edge} users as follows. Consider the largest disk centered at $B_m$ and contained inside the $m$th Voronoi cell, and represent this disk using $\mathcal{D}_m$. Specifically \cite{FossZuyev:VoronoiProcessPoisson:1996}
\begin{equation}\label{Eq:InnerDisk}
\mathcal{D}_m = \l\{Z \in \mathds{R}^2 \l| |Z - B_m| \leq \frac{1}{2}|B-B_m|\ \forall \ B \in \Omega\backslash\{B_m\} \r. \r\}.
\end{equation}
Using the above definition, the inner-cell and cell-edge users in the $m$th cell are separated depending on whether they lie inside or outside the disk. In other words, the sets of inner-cell and cell-edge users are  $\{U\mid U\in \mathcal{V}_m \cap \mathcal{D}_m\}$ and $\{U\mid U\in \mathcal{V}_m \cap \mathcal{D}_m^c\}$, respectively, where $\mathcal{D}_m^c=\mathds{R}^2/\mathcal{D}_m$.
Typically, direct links between cell-edge users and their serving base stations are severely attenuated by pass loss. As a result, direct transmissions from these users to base stations are potentially difficult due to the required large transmission power. Furthermore, such direct transmissions cause strong interference to nearby users and ad hoc receivers. For these reasons, the uplink transmissions of  cell-edge users are assumed to be assisted by relay stations near cell edges \cite{ZhangLau:CoopRelayNexGenWlssSys:2007}. For simplicity, it is assumed that by relay transmission the SIR outage probabilities of the cell-edge users are no larger than those of inner-cell users. Thereby it is sufficient to consider only inner-cell users in the analysis.

\subsection{Channel and Modulation}\label{Section:ChanModel}
The uplink spectrum is divided into $M$ frequency-flat sub-channels by using \emph{orthogonal frequency division multiplexing} (OFDM) \cite{GoldsmithBook:WirelessComm:05}. Each of the coexisting networks uses a subset or the full set of sub-channels, depending on the spectrum sharing methods discussed in Section~\ref{Section:Sys:Overlay}. In each network, a transmitter modulates signals using frequency-hopping spread spectrum, where signals hope randomly over all sub-channels assigned to the network \cite{SimonBook:SpreadSpectrumHandbook:94, WeberAndrews:TransCapWlssAdHocNetwkOutage:2005}.

Consider the link between a typical user and the serving base station, denoted as $U_0$ and $B_0$, respectively. A typical sub-channel accessed by $U_0$ consists of path loss and a fading factor denoted by $W$ such that the signal power received by $B_0$ is $\rho WD^{-\alpha}$, where $\rho$ is the transmission power and $D = |U_0-B_0|$. Similarly, the interference power from an interferer $X$ to $B_0$ is $P_X G_XR_X^{-\alpha}$, where $P_X\in\{\rho, \tilde{\rho}\}$,  $R_X = |X-B_0|$, and $G_X$ is the fading factor.

Similar channel models are used for the ad hoc network. Specifically, the received signal power for a typical receiver, denoted as $T_0$, is $\widetilde{W}\tilde{d}^{-\alpha}$ where $\widetilde{W}$ is the fading factor; the interference power from an interferer $X$ to $T_0$ is $P_XG_XR_X^{-\alpha}$ where $\tilde{R}_X = |X-T_0|$ and $G_X$ is the fading factor mentioned earlier.

\subsection{Spectrum Sharing Methods}\label{Section:Sys:Overlay}
For spectrum overlay, the $M$ sub-channels are divided into two disjoint subsets and assigned to two coexisting networks.\footnote{We assume that different cells use the identical sets of sub-channels. Without this assumption, the users and the ad hoc nodes accessing one particular sub-channel are \emph{non-homogeneous} PPPs. The analysis of this case is complicated and delegated to future work.} Let $K$ and $\tilde{K}$ denote the numbers of sub-channels used by the cellular and ad hoc networks respectively, where $K+\tilde{K}=M$. Spectrum overlay requires initialization, where the cellular network communicates to the MANET the indices of the available sub-channels and the allowable node density. One initialization  method is to use base stations to broadcast control signals to ad hoc nodes. The constraint on the node density can be satisfied by distributed adjustments of nodes' transmission probability, \emph{thinning} the PPP of ad hoc transmitters \cite{StoyanBook:StochasticGeometry:95}. Moreover, $K$ and $\tilde{K}$ can be adapted to the time-varying uplink traffic load, increasing spectrum-usage efficiency at the cost of additional initialization overhead. Next, for spectrum underlay, both coexisting networks use all $M$ sub-channels. Compared with spectrum overlay, spectrum underlay has less initialization overhead as the cellular network need not inform the ad hoc network the indices of available sub-channels.

The transmission capacities of the coexisting networks can be increased by employing SIC at each base station and ad hoc receiver for reducing interference. The SIC model is modified from that in \cite{WeberAndrews:TransCapWlssAdHocNetwkSIC:2005} for making tractable analysis of fading and network coexistence not considered in \cite{WeberAndrews:TransCapWlssAdHocNetwkSIC:2005}. For effective SIC, the SIC model in \cite{WeberAndrews:TransCapWlssAdHocNetwkSIC:2005} requires the interference power from each targeted interferer to be larger than the signal power, and furthermore the average number of canceled interferers is upper bounded. In this paper, by combining these two SIC constraints,  the interference power of each targeted interferer must exceed a threshold equal to the received signal power multiplied by a factor larger than one, denoted as $\kappa$. Increasing $\kappa$ decreases the average number of canceled interferers and vice versa. Finally, perfect SIC is assumed.

\subsection{Transmission Capacity}\label{Section:Sys:TxCap}
Network transmission capacities of the coexisting networks are defined in terms of outage probabilities \cite{WeberAndrews:TransCapWlssAdHocNetwkOutage:2005}. As in \cite{WeberAndrews:TransCapAdHocNetwkDistSch:2006}, the networks are assumed to be interference limited and thus noise is neglected for simplicity.
Consequently, the reliability of received data packets is measured by the SIR. Let $\SIR$ and $\tSIR$ represent the SIRs at the typical user $U_0$ and ad hoc receiver $T_0$, respectively. The correct decoding of received data packets requires the SIRs to exceed a threshold $\theta \geq 1$, identical for all receivers in the networks. In other words, the rate of information sent from a transmitter to a receiver is no less than $\log_2(1+\theta)$ assuming Gaussian signaling. To support this information rate with high probability, the outage probability that $\SIR$ and $\tSIR$ are below $\theta$ must be no larger than a given threshold $0<\epsilon\ll 1$, i.e.
\begin{equation}\label{Eq:Pout:Def}
\Pout(\lambda) := \Pr(\SIR < \theta) \leq \epsilon, \quad \tPout(\tilde{\lambda}) := \Pr(\tSIR < \theta) \leq \epsilon
\end{equation}
where $\Pout$ and $\tPout$ denote the SIR outage probabilities for the cellular and  the ad hoc networks, respectively. The transmission capacities of the cellular and the ad hoc networks, denoted as $C$ and $\tilde{C}$ respectively, are defined as \cite{WeberAndrews:TransCapWlssAdHocNetwkOutage:2005}
\begin{equation}\label{Eq:TxCap}
    C(\epsilon) = (1-\epsilon)\lambda_{\epsilon}, \quad     \tilde{C}(\epsilon) = (1-\epsilon)\tilde{\lambda}_{\epsilon}
\end{equation}
where $\lambda_{\epsilon}$ and $\tilde{\lambda}_{\epsilon}$ satisfy $\Pout(\lambda_{\epsilon}) = \epsilon$ and $\tPout(\tilde{\lambda}_{\epsilon}) = \epsilon$.

\section{Outage Probabilities}\label{Section:Outage}
In this section, the outage probabilities for the coexisting networks are derived for spectrum overlay and underlay with and without SIC.

\subsection{Existing Analytical Approach}\label{Section:Weber}
The analysis in the subsequent sections adopts an existing approach for analyzing the outage probability given a Poisson filed of interferers \cite{WeberAndrews:TransCapWlssAdHocNetwkOutage:2005, WeberAndrews:TransCapAdHocNetwkDistSch:2006, WeberAndrews:TransCapWlssAdHocNetwkSIC:2005, Huang:SpatialInterfCancel:2008, Huang:SpatialInterfCancel:2008}. Based on the network model in Section~\ref{Section:ChanModel}, the aggregate interference power at a receiver in the networks is known as  a \emph{power-law shot noise} process \cite{Lowen:PowerLawShotNoise:1990}. Analyzing outage probabilities require deriving the complementary cumulative density function (CCDF) of such a process, which, unfortunately, has no closed-form expression \cite{Lowen:PowerLawShotNoise:1990, WeberAndrews:TransCapAdHocNetwkDistSch:2006}. For this reason, the existing approach resolves to deriving bounds on  the CCDF as summarized below in the context of the coexisting networks using spectrum overlay.

Without loss of generality, assume that the typical user $U_0$ accesses the $m$th sub-channel. Let $\Pi_m$ represent the process of users using this  sub-channel. By the Marking Theorem~\cite{Kingman93:PoissonProc}, $\Pi_m$  can be shown to be a homogeneous PPP with the density $\lambda/K$. Furthermore, the interferer process $\Pi_m\backslash\{U_0\}$ is also a homogeneous PPP with the same density $\lambda/K$ according to Slivnyak's Theorem \cite{StoyanBook:StochasticGeometry:95}. Define the process of strong interferers for $U_0$ conditioned on the link realization $\{W=w, D =d\}$ as $\Sigma_S(w,d)  = \{X \in \Pi_m\backslash\{U_0\} \mid R_X^{-\alpha}G_X > wd^{-\alpha}\theta^{-1}\}$, where each interferer alone guaranteers the outage for $U_0$. Moreover, the process of weaker interferers is defined as $\Sigma_S^c(w,d) = (\Pi_m\backslash\{U_0\})/\Sigma_S^c(w,d)$.\footnote{Note that the processes $\Sigma_S(w,d)$ and $\Sigma_S^c(w,d)$ are independent as a property of the PPP.} Define the interference power of the weak interferers as $I^c_S(w,d) := \sum_{X\in\Sigma_S^c(w,d)}\rho G_XR_X^{-\alpha}$. Thus, $\Pout$ can be written as
\begin{equation}\label{Eq:Pout:a}
\begin{aligned}
&\Pout  &=& \E\[\Pr\(\Sigma_S(W,D) = \emptyset \mid W,D\)\Pr\(I_S^c(W,D) >  WD^{-\alpha}\theta^{-1}\mid W,D\)\]+\\
&&&\Pr\(\Sigma_S(W,D) \neq \emptyset \).
\end{aligned}
\end{equation}
Considering only the strong interferers leads to a lower bound on $\Pout$, denoted as $\Pout^l$
\begin{equation}
\Pout^l := \E\[\Pr\(\Sigma_S(W,D) \neq \emptyset \)\] = 1 - \E\[e^{-\E\[|\Sigma_S(W,D)|\]}\].\nn
\end{equation}
Let $\Pout^l(w,d)$ represent $\Pout^l$ conditioned on $\{W=w, D=d\}$.  The upper bound on $\Pout$, denoted as $\Pout^u$, is obtained by bounding the term $\Pr\(I_S^c(W,D) >  WD^{-\alpha}\theta^{-1}\)$ in \eqref{Eq:Pout:a} using Chebyshev's inequality
\begin{equation}
\Pr\(I_S^c(W,D) >  WD^{-\alpha}\theta^{-1}\) \leq \frac{\var(I_S^c(W,D))}{\(WD^{-\alpha}\theta^{-1}- \E\[I_S^c(W,D)\]\)^2},\quad \frac{WD^{-\alpha}\theta^{-1}}{\E\[I_S^c(W,D)\]}> 1.
\end{equation}

Using \cite[Theorem~2]{WeberAndrews:TransCapAdHocNetwkDistSch:2006} obtained following the above approach, the bounds on $\Pout$ and $\tPout$ for spectrum overlay are given in the following lemma.
\begin{lemma}[Spectrum Overlay] \label{Lem:Pout:FreqSep}For the coexisting networks based on spectrum overlay, the bounds on SIR outage probabilities are given as follows.
\begin{enumerate}
\item {\bf Cellular network}:
\begin{equation}
\E\[\Pout^l\(W, D, \frac{\lambda}{K}\)\] \leq \Pout(K, \lambda) \leq \E\[\Pout^u\(W, D, \frac{\lambda}{K}\)\]
\end{equation}
where
\begin{eqnarray}
\Pout^l(w, d, \lambda) &=& 1- \exp\l(- \zeta  \lambda w^{-\delta} d^2\right)\label{Eq:PoutLb:Exp}\\
\Pout^u(w, d, \lambda) &=& 1- \xi(w, d, \lambda) \exp\l(-\zeta \lambda w^{-\delta} d^2 \r)\label{Eq:PoutUb:Exp}\\
\xi(w, d, \lambda) &=& \left\{\begin{aligned}
&\left[1-\frac{\frac{\delta}{2-\delta}\zeta d^2w^{-\delta}\lambda }{\left(1-\frac{\delta}{1-\delta}\zeta d^2w^{-\delta}\lambda \right)^2}\right]^+, && \frac{\delta}{1-\delta}\zeta d^2w^{-\delta}\lambda < 1\\
&0, && \text{otherwise}
\end{aligned}\right.
\end{eqnarray}
and $\zeta := \pi \theta^{\delta}\E[G^\delta]$.

\item {\bf MANET}:
\begin{equation}
\E\[\Pout^l\(\widetilde{W}, \tilde{d}, \frac{\tilde{\lambda}}{\widetilde{K}}\)\] \leq \tPout(\tilde{K}, \tilde{\lambda}) \leq \E\[\Pout^u\(\widetilde{W}, \tilde{d}, \frac{\tilde{\lambda}}{\widetilde{K}}\)\]
\end{equation}
where $\Pout^l(\cdot,\cdot, \cdot)$ and $\Pout^u(\cdot,\cdot, \cdot)$ are given in \eqref{Eq:PoutLb:Exp} and \eqref{Eq:PoutUb:Exp}, respectively.

\end{enumerate}
\end{lemma}

\subsection{Outage Probabilities: Spectrum Underlay}
For spectrum underlay, the SIRs for the coexisting networks can be written as
\begin{eqnarray}
\text{(Cellular)} \quad \SIR &=& \frac{\rho WD^{-\alpha}}{\rho\sum_{X\in\Pi_m\backslash\{U_0\}}G_XR_X^{-\alpha} + \tilde{\rho}\sum_{X\in\tilde{\Pi}_m}G_XR_X^{-\alpha}}\label{Eq:SIR:FreqShare} \\
\text{(MANET)} \quad \tSIR &=& \frac{\tilde{\rho} \widetilde{W}\tilde{d}^{-\alpha}}{\rho\sum_{X\in\Pi_m}G_XR_X^{-\alpha} + \tilde{\rho}\sum_{X\in\tilde{\Pi}_m\backslash\{T_0\}}G_XR_X^{-\alpha}}.
\end{eqnarray}
Using \eqref{Eq:SIR:FreqShare}, the bounds on $\Pout$ for the cellular network are derived as follows. The parallel derivation for the MANET is similar and thus omitted for brevity. For the cellular network, all interferers for $U_0$ (including ad hoc transmitters and other users) can be grouped into a homogeneous \emph{marked PPP} \cite{Kingman93:PoissonProc} defined below, where a mark $P_X\in\{\rho, \tilde{\rho}\}$ is transmission power
\begin{equation}\label{Eq:MarkedPP}
\Upsilon = \l\{(X, P_X) \left| X \in  \Pi_m\cup\tilde{\Pi}_m\backslash\{U_0\}, P_X \in \{\rho, \tilde{\rho}\}\r. \r\}.
\end{equation}
The distribution of $\Upsilon$ is given in the following lemma.
\begin{lemma} \label{Lem:Dist:PPP}The point process $\Upsilon $ is a homogeneous marked PPP with the density $(\lambda + \tilde{\lambda})/M$, where the marks are i.i.d and have the following distribution function
\begin{equation}
P_T = \l\{\begin{aligned}
&P, && \text{w.p.} \ \frac{\lambda}{\lambda + \tilde{\lambda}}\\
&\tilde{P}, && \text{w.p.} \ \frac{\lambda}{\lambda + \tilde{\lambda}}.
\end{aligned}\right.
\end{equation}
\end{lemma}

\begin{proof} See Appendix~\ref{App:Dist:PPP}. \end{proof}
Using this lemma, the bounds on $\Pout$ are derived and given in the following proposition.
\begin{proposition} \label{Prop:Pout:FreqShare}[Spectrum Underlay] For the coexisting  networks based on spectrum underlay, the outage probabilities are bounded as follows.
\begin{enumerate}
\item {\bf Cellular network}:
\begin{equation}\label{Eq:PoutLb:FreqShare}
\E\[\Pout^l\(W, D, \frac{\lambda+\eta^{-\delta}\tilde{\lambda}}{M}\)\] \leq \Pout(\lambda, \tilde{\lambda}) \leq \E\[\Pout^u\(W, D, \frac{\lambda+\eta^{-\delta}\tilde{\lambda}}{M}\)\]
\end{equation}

\item {\bf MANET}:
\begin{equation}
\E\[\Pout^l\(\widetilde{W}, \tilde{d}, \frac{\eta^{\delta}\lambda+\tilde{\lambda}}{M}\)\] \leq \tPout(\lambda, \tilde{\lambda}) \leq \E\[\Pout^u\(\widetilde{W}, \tilde{d}, \frac{\eta^{\delta}\lambda+\tilde{\lambda}}{M}\)\]
\end{equation}

\end{enumerate}
where $\eta := \rho/\tilde{\rho}$, and $\Pout^l(\cdot,\cdot, \cdot)$ and $\Pout^u(\cdot,\cdot, \cdot)$  are defined in Lemma~\ref{Lem:Pout:FreqSep}.
\end{proposition}

\begin{proof}See Appendix~\ref{App:Pout:FreqShare}.
\end{proof}
Proposition~\ref{Prop:Pout:FreqShare} shows that the outage probability for each network depends on the transmitter densities of both networks. This coupling is due to spectrum underlay and the resultant mutual interference between the coexisting networks. As shown in Section~\ref{Section:TxCap}, such coupling may result in smaller transmission capacities for spectrum underlay than those for spectrum overlay. Moreover, Proposition~\ref{Prop:Pout:FreqShare} also shows that the outage probabilities for spectrum underlay depend on the transmission power ratio $\eta$. The effect of $\eta$ is also characterized in Section~\ref{Section:TxCap}.

Finally, the probability density function (PDF) of $D$ for an inner-cell user is given in the following lemma, which  is required for computing the bounds on $\Pout$ for different overlay methods. Recall the assumption that the outage probabilities of relay-assisted cell-edge users are no smaller than those of inner-cell users (cf. Section~\ref{Section:Sys:NetwkArch}). Thus, the PDF of $D$ for cell-edge users are unnecessary for our analysis.
\begin{lemma}\label{Lem:PDF:D}
The probability density function (PDF) of $D$ for an inner cell user is given as
\begin{equation}\label{Eq:PDF:D}
f_D(t) = -8\pi\lambda_b t \Ei(-4\pi\lambda_b t^2)
\end{equation}
where the exponential integral $\Ei(x) = \int_{-\infty}^x t^{-1}e^tdt$.
\end{lemma}

\begin{proof}See Appendix~\ref{App:PDF:D}.
\end{proof}
It can be observed from \eqref{Eq:PDF:D} that the key parameter of the PDF of $D$ is the density of base station $\lambda_b$. Intuitively,  increasing the density of base stations reduces the cell sizes and thus $D$ and vice versa.

\subsection{Outage Probabilities: Spectrum Sharing with SIC}\label{Section:Pout:SIC}
The SIRs for the coexisting networks employing SIC are obtained as follows. With SIC, the conditional interferer processes  for the typical user $U$ and ad hoc receiver $T_0$, denoted respectively as $\Sigma(w, d)$ and $\widetilde{\Sigma}(\widetilde{W})$, are defined as
\begin{eqnarray}
    \Sigma(w, d) &:=& \l\{\begin{aligned}
&\l\{X \in \Pi_m\backslash\{U_0\}\l|  G_X R_X^{-\alpha} \leq  \kappa wd^{-\alpha} \r.\r\}, &&\text{spectrum overlay}\\
&\{X \in \Pi_m\cup\tilde{\Pi}_m\backslash\{U_0\}\l |  P_XG_X R_X^{-\alpha} \leq \kappa \rho wd^{-\alpha} \r.\}, && \text{spectrum underlay}
\end{aligned}\r.
\nn\\
\widetilde{\Sigma}(\widetilde{W}) &:=& \l\{\begin{aligned}
&\{X \in \tilde{\Pi}_m\backslash\{T_0\}\mid  G_X R_X^{-\alpha} \leq \kappa \widetilde{W}\tilde{d}^{-\alpha} \}, &&\text{spectrum overlay}\\
&\{X \in \Pi_m\cup\tilde{\Pi}_m\backslash\{T_0\}\mid  P_X G_X R_X^{-\alpha} \leq \kappa \tilde{\rho}\widetilde{W}\tilde{d}^{-\alpha} \}, && \text{spectrum underlay}
\end{aligned}\r.
\nn
\end{eqnarray}
where the factor $\kappa$ determines the power threshold for qualifying interferers for SIC (cf. Section~\ref{Section:Sys:Overlay}).
Using the above definitions, the SIRs for the cellular and the ad hoc networks, denoted respectively as $\SIR$ and $\tSIR$, can be written as
\begin{eqnarray}
\!\!\!\!\!\!\!\!\text{(Spectrum overlay)}\quad \SIR(w,d) &\!\!\!\!=\!\!\!\!& \frac{wd^{-\alpha}}{\sum_{X\in\Sigma^{(3)}}G_XR_X^{-\alpha}}, \quad \tSIR(\widetilde{W}) = \frac{\widetilde{W}\tilde{d}^{-\alpha}}{\sum_{X\in\widetilde{\Sigma}^{(3)}}G_XR_X^{-\alpha}}\label{Eq:SIR:FreqSep:SIC}\\
\!\!\!\!\!\!\!\!\text{(Spectrum underlay)}\quad  \SIR(w,d) &\!\!\!\!=\!\!\!\!& \frac{\rho wd^{-\alpha}}{\sum_{X\in\Sigma}P_XG_XR_X^{-\alpha}}, \quad \tSIR(\widetilde{W}) = \frac{\tilde{\rho}\widetilde{W}\tilde{d}^{-\alpha}}{\sum_{X\in\widetilde{\Sigma}}P_XG_XR_X^{-\alpha}} \label{Eq:SIR:FreqShare:SIC}
\end{eqnarray}
where the distribution of $P_X$ is given in Lemma~\ref{Lem:Dist:PPP}.

The outage probabilities of the SIRs in \eqref{Eq:SIR:FreqSep:SIC} and \eqref{Eq:SIR:FreqShare:SIC} are given in the following proposition.
\begin{proposition}\label{Prop:Pout:SIC}For spectrum sharing with SIC, the bounds on outage probabilities $\Pout$ and $\tPout$ can be modified from their counterparts for the case of no SIC as given in Lemma~\ref{Lem:Pout:FreqSep} and Proposition~\ref{Prop:Pout:FreqShare} by replacing the functions $\tPout^l$ and $\tPout^u$  with $\hPout^l$ and $\hPout^u$ correspondingly, which are given as
\begin{eqnarray}
\hPout^l(w, d, \lambda) &=& 1- \exp\l(- \chi\zeta  \lambda d^2 w^{-\delta}\right)\label{Eq:PoutLb:Exp:a}\\
\hPout^u(w, d, \lambda) &=& 1- \xi(w, d, \lambda) \exp\l(-\chi\zeta \lambda w^{-\delta} d^2 \r)\label{Eq:PoutUb:Exp:a}
\end{eqnarray}
where $\chi := \(1-\theta^{-\delta}\kappa^{-\delta}\)$ and the function $\xi(w, d, \lambda)$ is given in Lemma~\ref{Lem:Pout:FreqSep}.
\end{proposition}

\begin{proof}
See Appendix~\ref{App:PoutBnds:SIC}.
\end{proof}
Note that \eqref{Eq:PoutLb:Exp:a} and \eqref{Eq:PoutUb:Exp:a} differ from respectively \eqref{Eq:PoutLb:Exp} and \eqref{Eq:PoutUb:Exp} only by the factor $\chi$.
The factor $\chi < 1$ represents the SIC advantage of reducing outage probabilities with respect to the case of no SIC ($\chi = 1$). Moreover, decreasing the SIC factor $\kappa$ reduces $\chi$ and thus outage probabilities. Nevertheless, $\kappa$ being too small may invalidate the assumption of perfect SIC. Specifically, small $\kappa$ implies small SIR for the process of decoding interference prior to its cancelation and potentially results in significant residual interference after SIC \cite{Andrews:InterfCancelCellularSys:ContempView:2005}.

\section{Network Capacity Trade-Off: Asymptotic Analysis}\label{Section:TxCap}
Using the results obtained in the preceding section, the trade-off between the transmission capacities of the coexisting networks, namely $C$ and $\tilde{C}$ as defined in \eqref{Eq:TxCap}, is characterized in the following theorem for small target outage probability $\epsilon \rightarrow 0$.

\begin{theorem}\label{Theo:Throughput}
For $\epsilon \rightarrow 0$, transmission capacities of the coexisting networks satisfy
\begin{equation}\label{Eq:Throughput}
\tilde{\mu}\tilde{C} + \mu C = \frac{M}{\phi}\epsilon + O\(\epsilon^2\)
\end{equation}
where  the weights $\mu$ and $\tilde{\mu}$ are given as\footnote{The subscripts $o$ and $u$ identify spectrum overlay and underlay, respectively}
\begin{equation}\label{Eq:CapTO:Wei}
\l\{
\begin{aligned}
&\tilde{\mu}_o =\zeta\E[\widetilde{W}^{-\delta}]\tilde{d}^2, &&\mu_o =\zeta\E[W^{-\delta}](8\pi\lambda_b)^{-1},&& \text{spectrum overlay}\\
&\tilde{\mu}_u = \tilde{\mu}_o\vee(\eta^{-\delta}\mu_o), & &\mu_u =(\eta^{\delta}\tilde{\mu}_o)\vee\mu_o,&& \text{spectrum underlay}
\end{aligned}\r.
\end{equation}
and $\phi$ depends on if SIC is used
\begin{equation}\label{Eq:CapTO:SIC}
\l\{
\begin{aligned}
&\phi =1,&& \text{no SIC}\\
&1- \theta^{-\delta}\kappa^{-\delta} \leq \phi \leq \tfrac{2}{2-\delta}- \theta^{-\delta}\kappa^{-\delta},&& \text{SIC}.
\end{aligned}\r.
\end{equation}
\end{theorem}

\begin{proof}See Appendix~\ref{App:PoutBnds:SIC}.
\end{proof}
Theorem~\ref{Theo:Throughput} shows that the trade-off between $C$ and $\tilde{C}$ follows a linear equation. Specifically, the slope at which $\tilde{C}$ increases with decreasing $C$ is $-\mu/\tilde{\mu}$, which depends on different network parameters as observed from \eqref{Eq:CapTO:Wei}. The results in Theorem~\ref{Theo:Throughput} are interpreted using several corollaries in the sequel.

To facilitate discussion, define an \emph{outage limited} network as one whose transmission capacity is achieved with the outage constraint being active. For instance, the cellular network is outage limited if $C = (1-\epsilon)\lambda_\epsilon$ with  $\Pout(\lambda_\epsilon) = \epsilon$. For spectrum overlay, both the coexisting networks are outage limited. Nevertheless, for spectrum underlay, it is likely that only one of the two networks is outage limited as explained shortly. As implied by the proof for Theorem~\ref{Theo:Throughput}, for spectrum underlay, both networks are outage limited only if $\mu_u = \tilde{\mu}_u$, where $\mu_u$ and $\tilde{\mu}_u$ are given in \eqref{Eq:CapTO:Wei}. Otherwise, $\mu_u > \tilde{\mu}_u$ correspond to only the cellular network being outage limited; $\mu_u < \tilde{\mu}_u$ indicates that only the MANET is outage limited.

Spectrum overlay is shown to be more efficient than spectrum underlay as follows. Define the \emph{capacity region} of the coexisting networks as the region enclosed by the capacity trade-off curve in \eqref{Eq:Throughput} and the positive  axes of the $C$-$\tilde{C}$ coordinates. This region contains all feasible combinations of transmission capacities of coexisting networks. Thus, the size of the capacity region measures the efficiency of the overlaid network. The capacity regions for spectrum overlay and underlay are compared in the following corollary.
\begin{corollary}\label{Cor:CapRegion}
For $\epsilon\rightarrow 0$, the capacity region for spectrum underlay is no larger than that for spectrum overlay. They are identical  if and only if the transmission-power ratio is chosen as
\begin{equation}\label{Eq:PowerRatio:Optim}
\eta = \(\frac{\mu_o}{\tilde{\mu}_o}\)^{\frac{1}{\delta}}
\end{equation}
where $\mu_o$ and $\tilde{\mu}_o$ are given in \eqref{Eq:CapTO:Wei}.
\end{corollary}
\begin{proof}See Appendix~\ref{App:CapRegion}.
\end{proof}
Corollary~\ref{Cor:CapRegion} shows that spectrum overlay is potentially more efficient than spectrum underlay due to network coupling for the latter. Specifically, the possibility that a network is not outage limited compromises the efficiency of spectrum underlay, which, however, can be compensated by setting $\eta$ as given in \eqref{Eq:PowerRatio:Optim}. This optimal value of $\eta$ ensures both networks are outage limited for the case of spectrum underlay.

The next corollary specifies the effects of several parameters on transmission capacities  of the coexisting networks.
\begin{corollary} For $\epsilon \rightarrow 0$, transmission capacities  vary with network parameters as follows.
\begin{enumerate}
\item {\bf Spectrum overlay}: $C$ increases \emph{linearly} with the base station density $\lambda_b$; $\tilde{C}$ increases \emph{inversely} with the ad hoc     transmitter-receiver distance  $\tilde{d}$.

\item {\bf Spectrum underlay}: If the cellular network is outage limited, both $C$ and $\tilde{C}$ increase \emph{linearly} with the base station density $\lambda_b$. Otherwise, both $C$ and $\tilde{C}$ increase \emph{inversely} with the ad hoc     transmitter-receiver distance  $\tilde{d}$.

\item For both spectrum sharing methods, $C$ and $\tilde{C}$ increase \emph{linearly} with $\epsilon$ and the number of sub-channels $M$, and \emph{inversely} with $\phi$ related to SIC.
\end{enumerate}
\end{corollary}

Finally, we analyze the transmission-capacity gains due to spatial \emph{diversity gains} contributed by multi-antennas \cite{GoldsmithBook:WirelessComm:05}. To obtain concrete results, the fading factors $W$ and $\widetilde{W}$ are assumed to follow the chi-squared distributions with the degrees of freedom $L$ and $\tilde{L}$ respectively, which are the \emph{diversity gains}. These fading distributions can result from using spatial diversity techniques such as beamforming over multi-antenna i.i.d. Rayleigh fading channels \cite{GoldsmithBook:WirelessComm:05, PaulrajBook}. Thus
\begin{equation}\label{Eq:Gamma}
\E[W^{-\delta}] = \frac{\Gamma(L - \delta)}{\Gamma(L)},\quad \E[\widetilde{W}^{-\delta}] = \frac{\Gamma(\tilde{L} - \delta)}{\Gamma(\tilde{L})}.
\end{equation}
The following corollary is obtained by combining Theorem~\ref{Theo:Throughput}, \eqref{Eq:Gamma} and the following Kershaw's Inequalities \cite{Kershaw:GautschiInequGammaFun:1983}
\begin{equation}\label{Eq:Kershaw}
\(x + \frac{s}{2}\)^{1-s} < \frac{\Gamma(x + s)}{\Gamma(x + 1)} < \[x-\frac{1}{2} + \(s+\frac{1}{4}\)^{\frac{1}{2}}\]^{1-s},\quad x\geq 1,\ 0 < s < 1.
\end{equation}

\begin{corollary}[Spatial Diversity Gain] Consider the diversity gains per link of $L$ and $\tilde{L}$ for the coexisting cellular and ad hoc networks, respectively.
\begin{enumerate}
\item {\bf Spectrum overlay}: The spatial diversity gains multiply $C$ by a factor between $(L-1)^\delta$ and $L^\delta$,  and $\tilde{C}$ by a factor between $(\tilde{L}-1)^\delta$ and $\tilde{L}^\delta$.

\item {\bf Spectrum underlay}: The spatial diversity gains multiply  both $C$ and $\tilde{C}$ by a factor between $(L-1)^\delta$ and $L^\delta$ if the cellular network is outage limited, or otherwise between $(\tilde{L}-1)^\delta$ and $\tilde{L}^\delta$.
\end{enumerate}
\end{corollary}
Note that similar results are obtained in \cite{AndrewJeff:CapacityScalingSpatialDiversity:2006} for a standing-alone MANET by using a more complicated method than the current one based on Kershaw's Inequalities.

\section{Simulation and Numerical Results}\label{Section:Numerical}
In this section, the tightness of the bounds on outage probabilities derived in Section~\ref{Section:Outage} is evaluated using simulation. Moreover, the asymptotic transmission capacity trade-off curves obtained in Theorem~\ref{Theo:Throughput} are compared with the non-asymptotic ones generated by simulation. The simulation procedure summarized below is similar to that in \cite{WeberKam:CompComplexMANETs:2006}. The typical base station (or the ad hoc receiver) of the coexisting network lies at the centers of two overlapping disks, which contain interfering transmitters (either ad hoc nodes, users or both) and base stations respectively. Both the transmitters and the base stations follow the Poisson distribution with the mean equal to $200$. The disk radiuses are adjusted to provide the desired densities of transmitters or base stations. For simulations, the distance between the typical ad hoc transmitter and receiver is $d=5$ m, the required SIR  $\theta = 3$ or $4.8$ dB, the path-loss exponent $\alpha = 4$, the base station density  $\lambda_b = 10^{-3}$, the SIC factor $\kappa = 2$ dB, and the transmission-power ratio $\eta = 5$ dB.

Fig.~\ref{Fig:PoutCmp} compares the bounds on outage probabilities in Section~\ref{Section:Outage} and the simulated values. As observed from Fig.~\ref{Fig:PoutCmp}, for all cases, the outage probabilities converge to their lower bounds as the transmitter densities decrease; the upper and lower bounds differ by approximately constant multiplicative factors. Fig.~\ref{Fig:PoutCmp} also shows that SIC reduces outage probabilities by a factor of about  $0.54$ approximately equal to $\chi$ in Proposition~\ref{Prop:Pout:SIC}. Moreover, SIC loosens the bounds on outage probabilities for relatively large transmitter densities since SIC reduces the number of strong interferers to each receiver. Finally, outage probabilities become proportional to transmitter densities as they  decrease.

Fig.~\ref{Fig:TxCap} compares the asymptotic transmission-capacity trade-off curves in Theorem~\ref{Theo:Throughput} and those generated by simulations for the target outage probability $\epsilon = 10^{-2}$.  In Fig.~\ref{Fig:TxCap}(b) for the case of SIC, the bounds on the asymptotic trade-off curves correspond to those on $\phi$ as given Theorem~\ref{Theo:Throughput}. By comparing Fig.~\ref{Fig:TxCap}(a)  and Fig.~\ref{Fig:TxCap}(b), the capacity regions for spectrum overlay are larger than those for spectrum underlay. For the case of no SIC, the asymptotic results closely match their simulated counterparts. When SIC is used, the capacity trade-off curves generated by simulation are close to the corresponding asymptotic upper bounds. In particular, for spectrum overlay with SIC, the simulation results are practically identical to their asymptotic upper bounds. In summary, the asymptotic results derived in Section~\ref{Section:TxCap} are  useful for characterizing the transmission capacities of the coexisting networks in the non-asymptotic regime.

\section{Conclusion}\label{Section:Conlusion}
In this paper,  the transmission-capacity trade-off between the coexisting cellular and ad hoc networks is analyzed for different spectrum sharing methods. To this end, bounds on outage probabilities for both networks are derived for spectrum overlay and underlay with and without SIC. For small target outage probability, the transmission capacities of the coexisting networks are shown to satisfy a linear equations, whose coefficients are derived for the cases considered above. These results provide a theoretical basis for adapting the node density of the ad hoc network to the dynamic of the traffic in cellular uplink under the outage constraint for both networks. The trade-off relationship suggests that transmission capacities of coexisting networks can be increased by adjusting various parameters such as decreasing the distances between intended ad hoc transmitters and receivers, increasing the base station density and link diversity gains, or by employing SIC. In particular, SIC increases the transmission capacities by a linear factor that depends on the interference power threshold for qualifying canceled interferers. Simulation results show that the derived bounds on outage probabilities are tight and the asymptotic liner capacity trade-off is valid even in the non-asymptotic regime.

This paper opens several issues for future work on spectrum sharing between networks including the impact of cognitive radio, the capacity trade-off between competing networks, and the extension to more realistic non-homogeneous network architectures.

\appendix
\subsection{Proof for Lemma~\ref{Lem:Dist:PPP}}\label{App:Dist:PPP}
By using the superposition property of Poisson processes, the combined PPP $\Pi_m\cup\tilde{\Pi}_m$ is also a homogeneous PPP with the density $\frac{\lambda+\tilde{\lambda}}{M}$. Consider a typical point $X\in \Pi_m\cup\tilde{\Pi}_m$. Let $B(A, r)$ denote  a disk centered at a point $A\in\mathcal{R}^2$ and with a radius $r$, thus $B(A, r) = \{X\in\mathds{R}^2\left| |X-A|\leq r \r.\}$. Moreover, the area of $B(A, r)$ is denoted as $\mathcal{A}(B(A, r))$. Thus the probability for the event that $X$ belongs to $\Pi_m$, or equivalently $P_X = \rho$, is
\begin{eqnarray}
\Pr(X\in\Pi_m) &=& \lim_{r \rightarrow 0} \frac{1-\exp\(\frac{\lambda}{M} \mathcal{A}(B(X, r))\)}{1-\exp\(\frac{\lambda+\tilde{\lambda}}{M} \mathcal{A}(B(X, r))\)}  \nn\\
&=& \lim_{r \rightarrow 0} \frac{\lambda\exp(\lambda \pi r^2/M) }{(\lambda+\tilde{\lambda})\exp((\lambda+\tilde{\lambda}) \pi r^2/M)} = \frac{\lambda}{\lambda + \tilde{\lambda}}.\nn
\end{eqnarray}
Similarly, $\Pr(X\in\tilde{\Pi}_m) = \frac{\tilde{\lambda}}{\lambda + \tilde{\lambda}}$. This completes the proof.

\subsection{Proof for Proposition~\ref{Prop:Pout:FreqShare}}\label{App:Pout:FreqShare}
The marked point process in  \eqref{Eq:MarkedPP} is modified to include the fading factor $G_X$ as an additional mark as follows
\begin{equation}\label{Eq:MarkedPP:a}
\hat{\Upsilon} := \l\{(X, P_X, G_X) \left| X \in  \Pi_m\cup\tilde{\Pi}_m\backslash\{U_0\}, P_X \in \{\rho, \tilde{\rho}\}, G_X\in\mathds{R}^+\r. \r\}.
\end{equation}
Following the approach discussed in Section~\ref{Section:Weber}, $\hat{\Upsilon}$ is divided into a strong-interferer sub-process conditioned on $(W=w, D=d)$, denoted as $\hat{\Upsilon}_S(w,d)$ and given as
\begin{equation}\label{Eq:MarkedPP:b}
\Upsilon_S(w,d) = \l\{(X, P_X, G_X) \left| (X, P_X, G_X)\in\Upsilon', P_X|X|^{-\alpha}G_X > \rho wd^{-\alpha}\theta^{-1}\r. \r\}
\end{equation}
and the weak-interferer process defined as $\hat{\Upsilon}_S^c(w,d) = \hat{\Upsilon}/\Upsilon_S(w,d)$. Thus, the sum interference power from weak interferers can be written as  $I_S^c(w,d) = \sum_{(X, P_X, G_X)\in\hat{\Upsilon}_S^c(w,d)}P_X|X|^{-\alpha}G_X$. To apply the analytical procedure in Section~\ref{Section:Weber}, it is sufficient to obtain $\E[|\hat{\Upsilon}_S(w,d)|]$, $\E[I_S^c(w,d)]$ and $\var\[I_S^c(w,d)\]$. Using the Marking Theorem \cite{Kingman93:PoissonProc} and Lemma~\ref{Lem:Dist:PPP},
\begin{eqnarray}
\E\[|\Upsilon_S(w,d)|\] &=& \frac{2\pi(\lambda + \tilde{\lambda})}{M}\[\Pr\(P_X = \rho\)\int_0^\infty\int_0^{(w^{-1}d^\alpha\theta g)^{\frac{1}{\alpha}}} r f_G(g) dr dg + \r.\nn\\
&& \l.\Pr\(P_X = \tilde{\rho}\)\int_0^\infty\int_0^{(\eta^{-1}w^{-1}d^\alpha\theta g)^{\frac{1}{\alpha}}} r f_G(g) dr dg\]\nn\\
&=& \frac{\zeta w^{-\delta}d^2(\lambda + \eta^{-\delta}\tilde{\lambda})}{M} \label{Eq:Density}
\end{eqnarray}
where $\zeta$ is defined in Lemma~\ref{Lem:Pout:FreqSep}. Next, $\E[I_S^c(w,d)]$ and $\var\[I_S^c(w,d)\]$ are derived using Campbell's Theorem \cite{Kingman93:PoissonProc} and Lemma~\ref{Lem:Dist:PPP} as follows
\begin{eqnarray}
\E[I_S^c(w,d)] &=& \frac{2\pi(\lambda + \tilde{\lambda})}{M}\[\Pr\(P_X = \rho\)\int_0^\infty\int_{(w^{-1}d^\alpha\theta g)^{\frac{1}{\alpha}}}^\infty (\rho r^{-\alpha} g ) r f_G(g) dr dg + \r.\nn\\
&& \l.\Pr\(P_X = \tilde{\rho}\)\int_0^\infty\int_{(\eta^{-1}w^{-1}d^\alpha\theta g)^{\frac{1}{\alpha}}}^\infty (\tilde{\rho} r^{-\alpha} g )r f_G(g) dr dg\]\nn\\
&=& \frac{\rho\delta}{1-\delta}\(\frac{\lambda+\eta^{-\delta}\tilde{\lambda}}{M}\)\zeta(w^{-1}d^\alpha)^{\delta-1}\theta^{-1},\label{Eq:Mean:Weak:a}\\
\var[I_S^c(w,d)] &=& \frac{2\pi(\lambda + \tilde{\lambda})}{M}\[\Pr\(P_X = \rho\)\int_0^\infty\int_{(w^{-1}d^\alpha\theta g)^{\frac{1}{\alpha}}}^\infty (\rho r^{-\alpha} g )^2 r f_G(g) dr dg + \r.\nn\\
&& \l.\Pr\(P_X = \tilde{\rho}\)\int_0^\infty\int_{(\eta^{-1}w^{-1}d^\alpha\theta g)^{\frac{1}{\alpha}}}^\infty (\tilde{\rho} r^{-\alpha} g )^2r f_G(g) dr dg\]\nn\\
&=& \frac{\rho^2\delta}{2-\delta}\(\frac{\lambda+\eta^{-\delta}\tilde{\lambda}}{M}\)\zeta(w^{-1}d^\alpha)^{\delta-2}\theta^{-2}.\label{Eq:Var:Weak:a}
\end{eqnarray}
Combining \eqref{Eq:Density}, \eqref{Eq:Mean:Weak:a}, \eqref{Eq:Var:Weak:a} and the analytical approach in Section~\ref{Section:Weber} gives the desired results.

\subsection{Proof for Lemma~\ref{Lem:PDF:D}}\label{App:PDF:D}
Let $Z$ denote the largest disk centered at a typical base station $B_0$ and contained inside the corresponding  Voronoi cell. Conditioned on $Z=z$, the CDF of $D$ of a typical inner-cell user is
\begin{equation}\label{Eq:PDF:D:Cond}
\Pr(D\leq t\mid Z=z) = \l\{\begin{aligned}
& 1,&& t\geq z\\
&\frac{t^2}{z^2},&& \text{otherwise.}
\end{aligned}\r.
\end{equation}
As a property of the random tessellation, the event $(Z \leq z)$ has the same probability as that where there is at least one other base station lying with in the distance of $2z$ from $B_0$ \cite{FossZuyev:VoronoiProcessPoisson:1996}. Mathematically
\begin{equation}\label{Eq:PDF:Z}
\Pr(Z\leq z) = 1 - e^{-4\pi\lambda_b z^2}.
\end{equation}
From \eqref{Eq:PDF:D:Cond} and \eqref{Eq:PDF:Z}
\begin{eqnarray}
\Pr(D\leq t) &=& \int_0^\infty \Pr(D\leq t\mid Z) f_Z(z) dz\nn\\
&=& \Pr(Z\leq t) + \int_t^\infty \frac{t^2}{z^2} \times 8\pi\lambda_b z e^{-4\pi\lambda_b z^2}dz\nn\\
&=& 8\pi\lambda_b t e^{-4\pi\lambda_b t^2} + 4\pi\lambda_bt^2\int_{4\pi\lambda_bt^2}^\infty z^{-1}e^{-z}dz.
\end{eqnarray}
Differentiating the above equation gives the desired result.

\subsection{Proof for Proposition~\ref{Lem:Pout:FreqSep}}\label{App:PoutBnds:SIC}
Only the bounds on $\Pout$ are proved. The proof for those on $\tPout$ is similar and thus omitted.

\subsubsection{Spectrum Overlay}
The interferers that are canceled at $B_0$ using SIC form a process defined as $\Sigma_C(w, d)  := \{X \in \Pi_m\backslash\{U_0\} \mid G_X D_X^{-\alpha}\geq  \kappa wd^{-\alpha}\theta^{-1}\}$. Define the process of strong interferers after SIC as $\Sigma_S(w,d) := \{X \in \Pi_m\backslash\{U_0\} \mid \theta^{-1}wd^{-\alpha} \leq G_T D_T^{-\alpha} \leq  \kappa wd^{-\alpha}  \}$. Note that $\kappa wd^{-\alpha}   > \theta^{-1}wd^{-\alpha} $ since $\kappa > 1$ and $\theta > 1$.  Thus, the process of weak interferers can be defined as $\Sigma_S^c(w,d) := (\Pi_m\backslash\{U_0\})/[\Sigma_S(w,d)\cup\Sigma_C(w,d)]$, which is observed to be identical to the counterpart for the case of no SIC. Since $\Sigma_S^c(w,d)\cap\Sigma_S(w,d) = \emptyset$, $\Sigma_S^c(w,d)$ and $\Sigma_S(w,d)$ are independent processes. From the discussion in Section~\ref{Section:Weber}, the exponential terms in \eqref{Eq:PoutLb:Exp} and \eqref{Eq:PoutUb:Exp} depends only on $\Sigma_S(w,d)$, and the function $\xi(w, d, \lambda/K)$ only on $\Sigma_S^c(w,d)$.  Since $\Sigma_S^c(w,d)$ is invariant to SIC, and $\Sigma_S^c(w,d)$ and $\Sigma_S(w,d)$ are independent, the bounds on $\Pout$ in Lemma~\ref{Lem:Pout:FreqSep} can be extended to the case of SIC by replacing the exponential term in \eqref{Eq:PoutLb:Exp} and \eqref{Eq:PoutUb:Exp}  with $\exp(-\E\[|\Sigma_S(w,d)|\])$, where $\E\[|\Sigma_S(w,d)|\]$ is obtained using Campbell's Theorem
\begin{equation}
\E\[\l|\Sigma_S(w, d)\r|\] = 2\pi\lambda\int_0^\infty\int_{(\kappa^{-1}w^{-1}d^\alpha g)^{\frac{1}{\alpha}}}^{(\theta w^{-1}d^\alpha g)^{\frac{1}{\alpha}}} rf_G(g)drdg = \chi \zeta w^{-\delta}d^2\frac{\lambda}{K}
\end{equation}
and $\chi$ is defined in the statement of the proposition.

\subsubsection{Spectrum Underlay}
With SIC, the strong and weak interferer process for $U_0$ are defined as $\widehat{\Sigma}_S(w,d) := \{X \in \Pi_m\backslash\{U_0\} \mid \theta^{-1}wd^{-\alpha} <  P_XG_X D_X^{-\alpha} \leq  \kappa wd^{-\alpha}  \}$ and $\widehat{\Sigma}^c_S(w,d) := \{X \in \Pi_m\backslash\{U_0\} \mid P_XG_X D_X^{-\alpha} \leq  \theta^{-1}wd^{-\alpha} \}$, respectively, where the distribution of $P_X$ is given in Lemma~\ref{Lem:Dist:PPP}. Based on the same arguments in the preceding section, the bounds on $\Pout$ in \ref{Eq:PoutLb:FreqShare} can be extended to the case of SIC by replacing their exponential terms with $\exp\(-\E\[|\widehat{\Sigma}_S(w,d)|\]\)$, where $\E\[|\widehat{\Sigma}_S(w,d)|\]$ is obtained using Campbell's Theorem as follows
\begin{eqnarray}
\E\[|\widehat{\Sigma}_S(w,d)|\] &=& \frac{2\pi(\lambda + \tilde{\lambda})}{M}\[\Pr\(P_X = \rho\)\int_0^\infty\int_{(\kappa^{-1}w^{-1}d^\alpha g)^{\frac{1}{\alpha}}}^{(w^{-1}d^\alpha\theta g)^{\frac{1}{\alpha}}} r f_G(g) dr dg + \r.\nn
\end{eqnarray}
\begin{eqnarray}
&& \l.\Pr\(P_X = \tilde{\rho}\)\int_0^\infty\int_{(\kappa^{-1}w^{-1}d^\alpha g)^{\frac{1}{\alpha}}}^{(\eta^{-1}w^{-1}d^\alpha\theta g)^{\frac{1}{\alpha}}} r f_G(g) dr dg\]\nn\\
&=& \frac{\chi \zeta w^{-\delta}d^2(\lambda + \eta^{-\delta}\tilde{\lambda})}{M}.\nn
\end{eqnarray}

\subsection{Proof for Theorem~\ref{Theo:Throughput}}\label{App:Throughput}
\subsubsection{Spectrum Overlay} The convergence $\epsilon \rightarrow 0$ implies  $\lambda\rightarrow 0$ and $\tilde{\lambda} \rightarrow 0$.
Using the series representation  of the PDF of a power shot-noise process \cite{Lowen:PowerLawShotNoise:1990}, the asymptotes of the outage probabilities follow from \cite[Theorem~2]{WeberAndrews:TransCapAdHocNetwkDistSch:2006}
\begin{equation}
\Pout = \lambda\zeta \E\[W^{-\delta}\]\E\[D^2\] + O\(\lambda^2\),\quad \tPout = \tilde{\lambda}\zeta \E\[\widetilde{W}^{-\delta}\]\tilde{d}^2 + O\(\tilde{\lambda}^2\). \label{Eq:Pout:c}
\end{equation}
By using \eqref{Eq:PDF:D:Cond} and \eqref{Eq:PDF:Z}, the term $\E\[D^2\]$ in \eqref{Eq:Pout:c} is obtained as follows
\begin{equation}
\E\[D^2\] = \E\[\int_0^z t^2 f_D(t\mid Z)dt\] = \E\[\frac{Z^2}{2}\] = \int_0^\infty \frac{z^2}{2} \times 8\pi\lambda_b ze^{-4\pi\lambda_b z^2}dz = \frac{1}{8\pi\lambda_b}. \label{Eq:ED2}
\end{equation}
Combining \eqref{Eq:TxCap}, \eqref{Eq:Pout:c}, and \eqref{Eq:ED2} gives the desired asymptotic capacity trade-off function for spectrum overlay.

\subsubsection{Spectrum Underlay}
By using the series expression  of the PDF of the power shot noise \cite{Lowen:PowerLawShotNoise:1990} as well as Proposition~\ref{Prop:Pout:FreqShare},
\begin{eqnarray}
\Pout(\lambda, \tilde{\lambda}) &=& \frac{\lambda + \eta^{-\delta}\tilde{\lambda}}{M}\zeta\E[W^{-\delta}]\E[D^2]  +O(\max(\lambda^2, \tilde{\lambda}^2))\label{App:Pout:a}\\
\tPout(\lambda, \tilde{\lambda}) &=& \frac{\eta^\delta\lambda + \tilde{\lambda}}{M}\zeta\E[\widetilde{W}^{-\delta}]\tilde{d}^2 + +O(\max(\lambda^2, \tilde{\lambda}^2)).\label{App:Pout:b}
\end{eqnarray}
For $\epsilon\rightarrow 0$, the transmission capacities $C$ and $\tilde{C}$ satisfy the constraints $\Pout(C/M, \tilde{C}/M)\leq \epsilon$ and $\tPout(C/M, \tilde{C}/M)\leq \epsilon$. By combining these constraints, \eqref{App:Pout:a} and \eqref{App:Pout:b}
\begin{eqnarray}
\frac{C + \eta^{-\delta}\tilde{C}}{M}\zeta\max\(\E[W^{-\delta}]\E[D^2], \eta^\delta\E[\widetilde{W}^{-\delta}]\tilde{d}^2\)  =\epsilon + O(\epsilon^2).
\end{eqnarray}
The desired result follows from the above equation.

\subsubsection{Spectrum Sharing with SIC}
Consider spectrum overlay with SIC. By canceling the strongest interferers using SIC, the PDF ``upper-tail" of the  power shot noise process is trimmed and its series expansion is difficult to find \cite{Lowen:PowerLawShotNoise:1990}. Nevertheless, the asymptotic transmission capacities can be characterized by expanding the bounds on $\Pout$ in Proposition~\ref{Prop:Pout:SIC}. Specifically
\begin{eqnarray}
\Pout^l(\lambda/K) &=&  \frac{\lambda}{K}\hat{\zeta}\E[W^{-\delta}]\E[D^2] + O(\lambda^2)\nn\\
\Pout^u(\lambda/K) &=&  1 - \E\[\(1-\frac{\delta}{2-\delta}\zeta W^{-\delta}D^2\frac{\lambda}{K} + O(\lambda^2)\)\(\frac{\lambda}{K}\hat{\zeta}W^{-\delta}D^2 + O(\lambda^2)\)\]\nn\\
&=& \(\frac{2}{2-\delta} - \theta^{-\delta}\kappa^{-\delta}\)\zeta\E[W^{-\delta}]\E[D^2]\frac{\lambda}{K}+O(\lambda^2).
\end{eqnarray}
Thus
\begin{equation}\label{Eq:PoutAsym}
\Pout(\lambda/K) = \chi \zeta\E[W^{-\delta}]\E[D^2]\frac{\lambda}{K}
\end{equation}
where $\(1 - \theta^{-\delta}\kappa^{-\delta}\)\leq \chi \leq \(\frac{2}{2-\delta} - \theta^{-\delta}\kappa^{-\delta}\)$. Similarly
\begin{equation}\label{Eq:tPoutAsym}
\tPout(\tilde{\lambda}/\tilde{K}) = \chi \zeta\E[\widetilde{W}^{-\delta}]\tilde{d}^2\frac{\tilde{\lambda}}{\tilde{K}}.
\end{equation}
The desired results for spectrum overlay with SIC are obtained by combining \eqref{Eq:TxCap}, \eqref{Eq:PoutAsym}, and \eqref{Eq:tPoutAsym}. The results for spectrum underlay with SIC are derived following a similar procedure.

%

\subsection{Proof for Corollary~\ref{Cor:CapRegion}}\label{App:CapRegion}
First, the capacity region for spectrum underlay is proved to be no larger than for spectrum overlay. It is sufficient to prove that $\mu_u\geq \mu_o$ and $\tilde{\mu}_u \geq \tilde{\mu}_o$, which follow from \eqref{Eq:CapTO:Wei}. Next,  substituting \eqref{Eq:PowerRatio:Optim} into \eqref{Eq:CapTO:Wei} results in $\mu_u = \mu_o$ and $\tilde{\mu}_u =  \tilde{\mu}_o$. This proves the second claim in the theorem statement.

\renewcommand{\baselinestretch}{1.2}
\bibliographystyle{ieeetr}

\newpage
\begin{figure}
\centering
  \includegraphics[width=12cm]{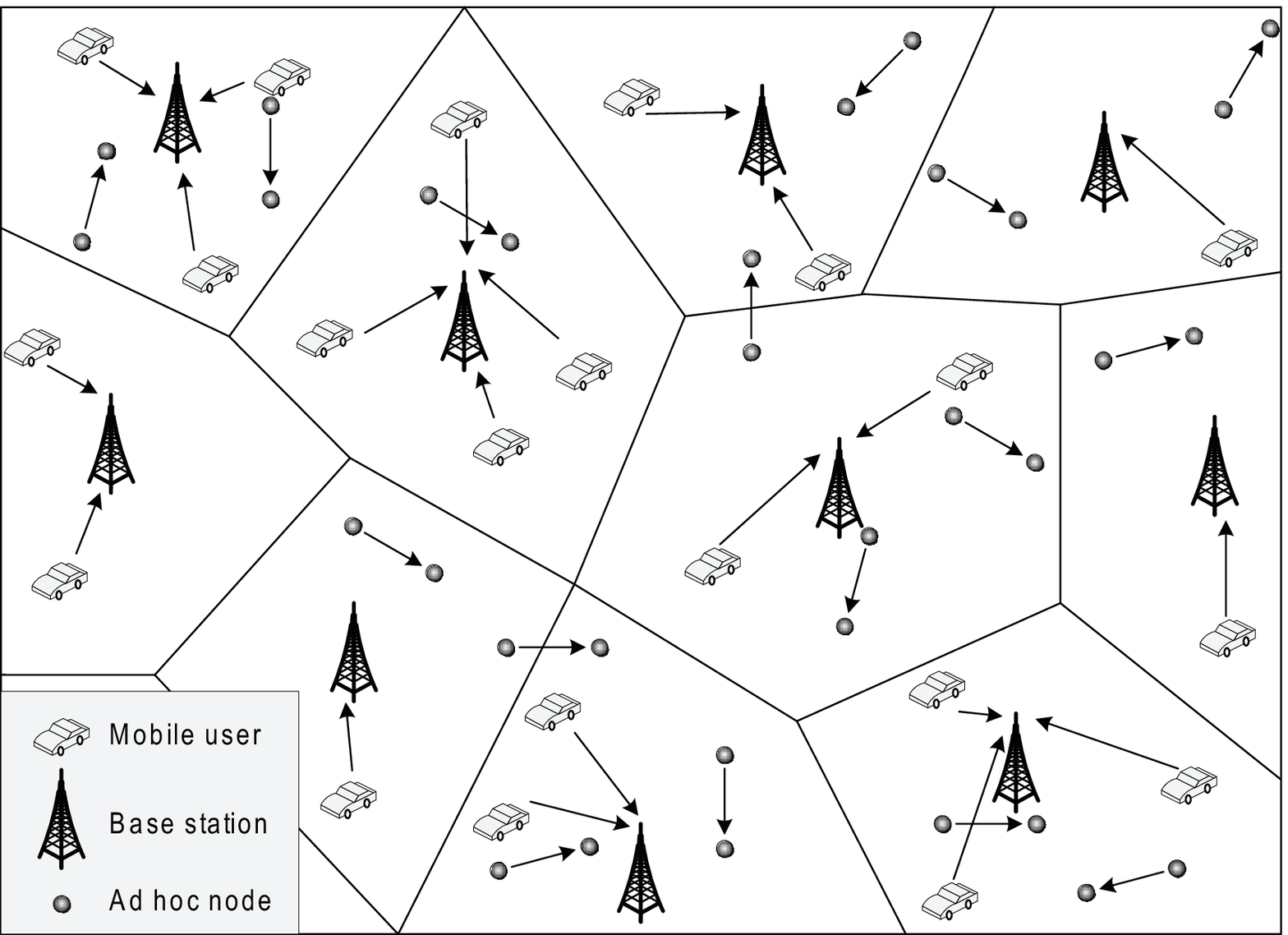}\vspace{10pt}\\
  \caption{The coexisting cellular and ad hoc networks}\label{Fig:Network}
\end{figure}

\begin{figure}
\centering
\subfigure[Spectrum overlay: cellular network]{  \includegraphics[width=8cm]{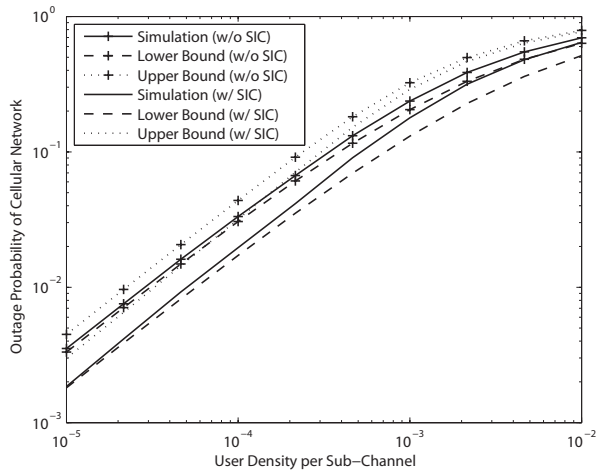}}
\subfigure[Spectrum overlay: ad hoc Network]{  \includegraphics[width=8cm]{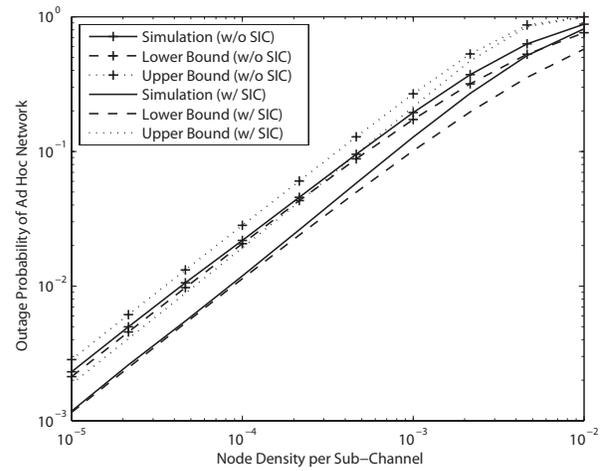}}\\
\subfigure[Spectrum underlay: cellular network]{  \includegraphics[width=8cm]{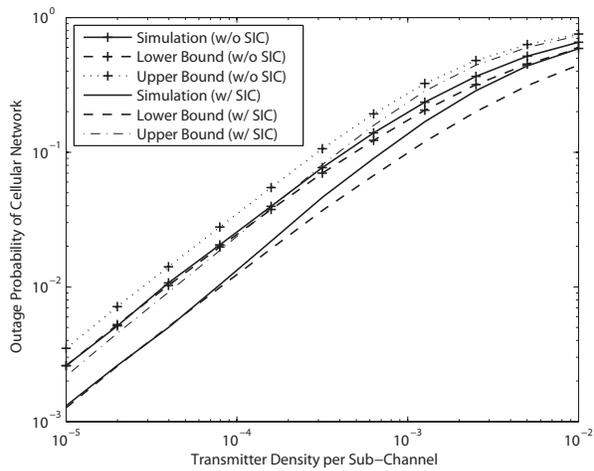}}
\subfigure[Spectrum underlay: ad hoc Network]{  \includegraphics[width=8cm]{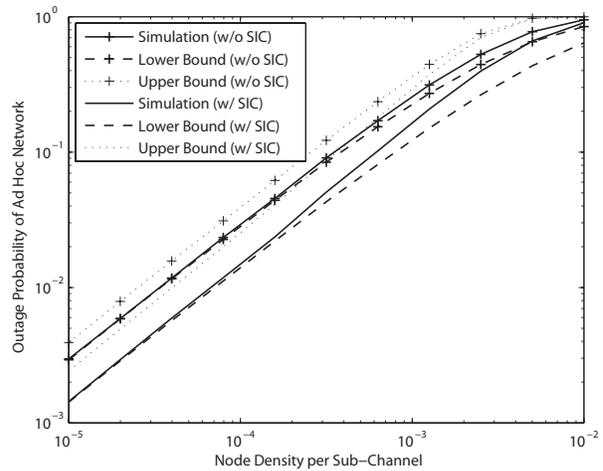}}\\
  \caption{Comparison between the theoretical bounds on outage probabilities and the simulated values. For spectrum underlay, the densities of users and ad hoc transmitters are set  equal, corresponding to one operational point on their trade-off curve. The sum density is referred to in the figures as the transmitter density.} \label{Fig:PoutCmp}
\end{figure}

\begin{figure}
\centering
\subfigure[Spectrum overlay]{  \includegraphics[width=12cm]{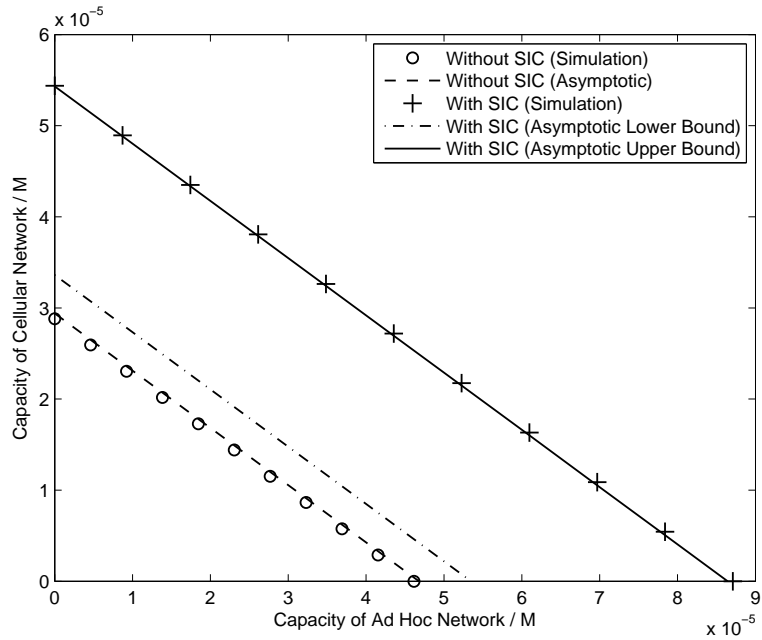}}\\
\subfigure[Spectrum underlay]{  \includegraphics[width=12cm]{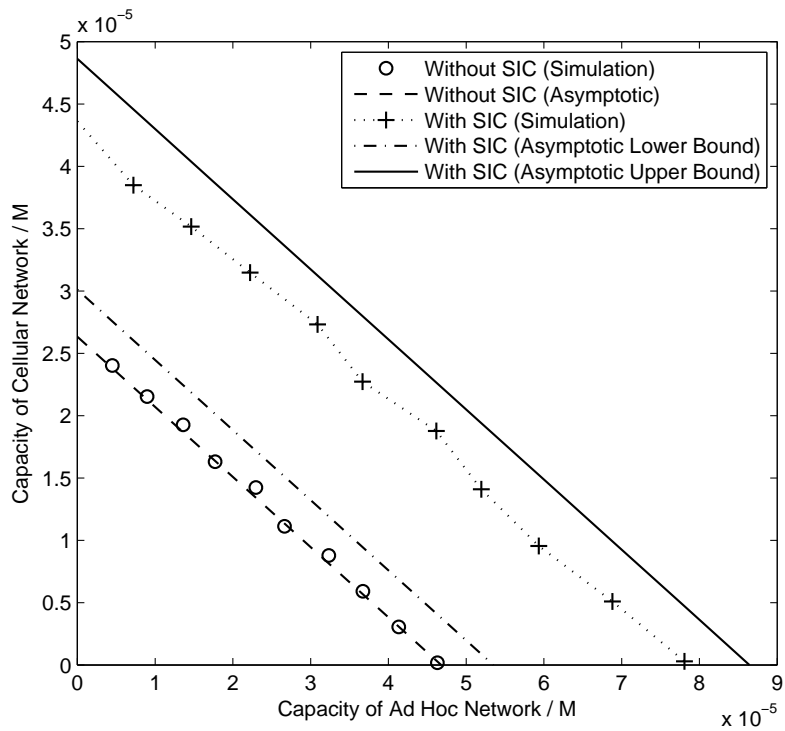}}\\
  \caption{Comparison between the asymptotic and the simulated transmission-capacity trade-off curves for the coexisting networks using (a) spectrum overlay or (b) spectrum underlay}\label{Fig:TxCap}
\end{figure}

\end{document}